\documentclass[11pt]{article}
\usepackage{times}
\usepackage{fullpage}
\usepackage{epsf}
\usepackage{amsmath,amsfonts}
\newtheorem{proposition}{proposition}[section]
\newtheorem{lemma}[proposition]{Lemma}

\newtheorem{theorem}[proposition]{Theorem}

\newtheorem{observation}[proposition]{Observation}

\newtheorem{definition}[proposition]{Definition}
\newcommand{\DEF}[1]{{\em #1\/}}

\newcommand\ol{\overline}
\newcommand\req[1]{(\ref{#1})}

\newcommand{\qed}{\hbox{\rule{6pt}{6pt}}}
\newenvironment{proof}[1][]{\paragraph{Proof{#1}}}{\hfill\qed\medskip\\}
\newcommand\C{{\mathcal C}}

\def\showlabel#1{}
\def\showfiglabel#1{}
\def\date{\today}  

\def\junk#1{}
\def\journal#1{}
\begin{document}
\setcounter{page}0

\title{Minimum $k$-way cut of bounded size is fixed-parameter tractable}

\author{{\em Ken-ichi Kawarabayashi}\thanks{Research partly
supported by Japan Society for the Promotion of Science,
Grant-in-Aid for Scientific Research, by C \& C Foundation, by
Kayamori Foundation and by Inoue Research Award for Young
Scientists.}\\National Institute of Informatics\\2-1-2 Hitotsubashi,
Chiyoda-ku\\ Tokyo 101-8430, Japan\\ \texttt{k\_keniti@nii.ac.jp}
 \and {\em Mikkel Thorup}\\
AT\&T Labs---Research\\
180 Park Avenue, Florham Park, NJ 07932, USA\\
\texttt{mthorup@research.att.com}}

\maketitle

\begin{abstract}
We consider a the minimum $k$-way cut problem for unweighted graphs
with a size bound $s$ on the number of cut edges allowed. Thus we
seek to remove as few edges as possible so as to split a graph into
$k$ components, or report that this requires cutting more than $s$
edges. We show that this problem is fixed-parameter tractable (FPT)
in $s$.  More precisely, for $s=O(1)$, our algorithm runs in
quadratic time while we have a different linear time algorithm for
planar graphs and bounded genus graphs.

Our tractability result stands in contrast to known W[1] hardness of
related problems. Without the size bound, Downey et al.~[2003] proved
that the minimum $k$-way cut problem is W[1] hard in $k$ even for
simple unweighted graphs. Downey et al.~asked about the status for
planar graphs. Our result implies tractability in $k$ for the planar
graphs since the minimum $k$-way cut of a planar graph is of size at
most $6k$ (in fact, the size is $f(k)$ for any bounded degree graphs
for some fixed function $f$ of $k$.  This class includes bounded genus
graphs, and simple graphs with an excluded minor).

A simple reduction shows that vertex cuts are at least as hard as edge
cuts, so the minimum $k$-way vertex cut is also W[1] hard in terms of
$k$. Marx [2004] proved that finding a minimum $k$-way vertex cut of
size $s$ is also W[1] hard in $s$. Marx asked about the FPT status
with edge cuts, which we prove tractable here. We are not aware of any
other cut problem where the vertex version is W[1] hard but the edge
version is FPT.
\end{abstract}


\thispagestyle{empty}
\clearpage

\journal{{\bf Keywords} : $k$-way-cut, planar graphs, bounded genus graphs
and FPT.}



\section{Introduction}
We consider the {\em minimum $k$-way cut problem\footnote{There is a lot
of confusing terminology associated with cut problems, e.g., in
the original conference version of \cite{DJPSY94}, ``multiway cut'' referred to the
separation of given terminals, but that term is fortunately corrected to
``multiterminal cut'' in the final journal version.
Here we follow the latter more explicit terminology: $k$-way cut for arbitrary
splitting into $k$ pieces,
$k$-terminal cut for splitting $k$ terminals, $k$-pair cut for splitting $k$
pairs, and so forth...}} of graph. The
goal is to find a minimum set of {\em cut\/} edges so as
to split the graph into at least $k$ components.
If a given graph is unweighted, minimum means minimum cardinality; otherwise
it means minimum total weight. Goldscmidt and Hochbaum \cite{GH94}
proved that the problem is NP-hard when $k$ is part of the input but
solvable in polynomial time for any fixed $k$.
Finding a minimum $k$-way cut is
an extension of the classical minimum cut problem, and it has applications
in the area of VLSI system design, parallel computing systems,
clustering, network reliability and finding cutting planes for the
traveling salesman problem.

Our focus is the minimum $k$-way cut problem for an
unweighted graph, deciding if there is a $k$-way cut of size $s$.  For
constant $s$ we solve this problem in quadratic time.
In the case of weights, our algorithms generalize to
finding the minimum weight $k$-way cut with at most $s$ edges.

For planar and, more generally, bounded genus graphs, we present a
different linear time algorithm for bounded size minimum $k$-way cut.
For simple unweighted bounded genus
graphs, we know that a minimum $k$-way cut has size $\Theta(k)$, so we
get linear time whenever $k=O(1)$.

Our result implies that the $k$-way cut problem is fixed-parameter
tractable when parameterized by the cut size $s$. Recall here that fixed-parameter tractable
(FPT) in a parameter $t$ means that there is an algorithm with running
time $O(f(t)n^c)$ for some fixed function $f$ and constant $c$. In our
case we get $t=s$, $c=2$, and $f(t)=t^{t^{O(t)}}$ for general graphs.
For bounded genus graphs, we get $t=s$, $c=1$ and $f(t)=2^{O(t^2)}$.
If the bounded genus graphs are simple and unweighted, we can
also use $t=k$ as parameter and get the same asymptotic bounds.

\begin{table}[t]
\begin{center}
\begin{tabular}{l|cc}
& parametrized by $k$ & parametrized by size $s$ \\\hline
$k$-way vertex cut of size $s$ & W[1] hard \cite{DEFPR03} &
W[1] hard \cite{Mar06} \\
$k$-way edge cut of size $s$ & W[1] hard \cite{DEFPR03} & FPT [This paper]
\end{tabular}
\caption{FPT status of $k$-way cut problems}\label{tab:FPT}
\end{center}
\end{table}

Our FPT result stands in contrast to known W[1] hardness of related
problems (c.f. Table \ref{tab:FPT}). Recall that if a problem is
W[1] hard in $t$, then it is not FPT in $t$ unless NP$=$P.
Without the size bound, Downey et
al.~\cite{DEFPR03} proved that the minimum $k$-way cut problem is W[1]
hard in $k$ even for simple unweighted graphs. Downey et al.~\cite{DEFPR03}
asked about the status for planar graphs. Our result implies tractability in
$k$ for the planar graphs since the minimum $k$-way cut of a planar
graph is of size at most $6k$.
Vertex cuts are at least as hard as edge cuts, so the minimum
$k$-way vertex cut is also W[1] hard in terms of $k$. Marx
\cite{Mar06} proved that finding a minimum $k$-way vertex cut of
size $s$ is also W[1] hard in $s$. Marx \cite{Mar06} asked about the
FPT status with bounded size edge cuts, which we prove tractable
here.

The discovered difference in FPT status between the edge and the
vertex version of the bounded size $k$-way cut problem is unusual
for cut problems. As mentioned above, both versions are W[1] hard when
only $k$ is bounded. On the other hand, Marx \cite{Mar06} has proved
that the bounded size $k$-terminal cut problem is FPT both for vertex
and edge cuts. He also proved FPT for a bounded size cut of a bounded
number of pairs.  Recently this was strenthened by Marx and Razgon
\cite{MR10} and Bousquet et al. \cite{BDT10}, showing that finding a
bounded size cut of an unbounded set of pairs is FPT both for vertex
and edge cuts.

Henceforth, unless otherwise specified, cuts are
understood to be edge cuts.

\subsection{More history}
\paragraph{General graphs.}
Goldschmidt and Hochbaum \cite{GH94} proved that
finding a minimum $k$-way cut is NP-hard when $k$ is part of
the input, but polynomial time solvable for fixed $k$. Their algorithm
finds a minimum $k$-way cut in $O(n^{(1/2-o(1))k^2})$ time and works
for weighted graphs. Karger and Stein \cite{KS96} proposed an
extremely simple randomized Monto Carlo algorithm for the
$k$-way cut problem whose running time is $O(n^{(2-o(1))k})$.  Then
Kamidoi et al.~\cite{KYN06} presented a deterministic algorithm that
runs in $O(n^{(4+o(1))k})$ time, and finally, Thorup \cite{Tho08}
presented the current fastest deterministic algorithm with
a running time of $\tilde O(m n^{2k-2})$.

The obvious big target would be to move the dependence on $k$ from
the exponent of $n$. However, this is impossible due to the
above mentioned W[1] hardness in $k$ by Downey et al. \cite{DEFPR03}.

As alternative to $k$, a very natural parameter to look at is the
cut size $s$, i.e., the size of the desired output. Getting
polynomial time for fixed $s$ is trivial since we can try all
subsets of $s$ edges in $O(n^{2s})$ time. Reducing this to
$O(n^{s})$ time straightforward using the sparsification from
\cite{NI92}. The challenge here is if we can move $s$ from the
exponent of $n$ and get FPT in $s$. As mentioned above, in the case
of vertex cuts, the bounded cut size was considered by Marx
\cite{Mar06} who proved the $k$-way vertex cut to be W[1] hard in
the size $s$. He asked if the edge version was also hard. Here we
show that the edge version is tractable with a quadratic algorithm
for any fixed cut size $s$.
Our algorithm implies that the minimum
$k$-way cut problem is solvable in polynomial time for bounded
degree unweighted graphs. This class includes planar graphs, bounded
genus graphs, and simple graphs with an excluded minor. In fact, we
give a linear time algorithm for planar graphs and bounded genus
graphs. Let us see more precisely.

\paragraph{Planar graphs.}
The special case of planar graphs has been quite well-studied.  In the
case of weighted planar graphs, Dahlhaus et al. \cite{DJPSY94} solved
the $k$-way cut problem in $O(n^{3k-1}\log n)$ time. The bound was
later improved to $\tilde O(n^{2k-1})$ time by Hartvigsen
\cite{Har93}. This was, however, matched by the later $\tilde O(m
n^{2k-2})$ bound by Thorup \cite{Tho08} for general graphs.

The case of simple unweighted planar graphs has also received
attention. Hochbaum and Shmoys \cite{HS85} gave an $O(n^2)$ algorithm
for simple unweighted planar graphs when $k=3$. This was
improved by He \cite{He91} to $O(n\log n)$. The motivation given in
\cite{HS85,He91} is the case of general $k$, with $k=3$ being the
special case for which they provide an efficient solution.

As mentioned previously, having proved the $k$-way cut problem W[1]
hard for simple unweighted general graphs, Downey et
al. \cite{DEFPR03} asked about the FPT status for planar graphs.

We resolve the question from \cite{DEFPR03} with an $O(2^{k^2}n)$
algorithm for simple unweighted planar graphs.  Even in the special
case of $k=3$, our result improves on the above mentioned algorithms
of Shmoys and Hochbaum \cite{HS85} and He \cite{He91}.

Our planar algorithm is generalized to the bounded genus case,
where it runs in time $O(2^{O(g^2k^2)}n)$.

\paragraph{Techniques.}
Our main result, the quadratic algorithm for general graphs, is a
simple combinatorial algorithm not relying on any previous results.
To solve the problem recursively, we will define ''the powercut
problem", which is much stronger than the minimum $k$-way cut
problem. It also generalizes the muli-pair cut problem with $p$ pairs for
fixed $p$ (this pair problem is, however, known to be FPT both for
vertices and edge cuts \cite{Mar06}). The approach can be seen as a typical example that a stronger inductive
hypothesis gives a much simpler inductive proof. With the  powercut problem,
it is easy to handle all high degree vertices except for one, which acts like
an apex vertex in graph minor theory. The rest is a bounded degree graph in
which we identify a contractible edge.

For planar graphs, like previous algorithms \cite{HS85,He91}, we
exploit that a minimum $k$-way cut has size $O(k)$. Otherwise our
algorithm for the planar case is based on a decomposition lemma from
Klein's \cite{Kle08} approximate TSP algorithm. In fact, this
appears to be the simplest direct application of Klein's lemma for a
classic problem. Klein's lemma is related to Baker's layered
approach \cite{Bak94} to planar graphs. Whereas Baker deletes layers
to get bounded tree-width, Klein contracts layers (deletion in the
dual) and such a contraction does not affect cuts avoiding the
contracted layers.  We present a linear time version of this
approach for bounded genus graphs.  In doing so, we also get a
linear time approximate TSP algorithm, improving an $O(n\log n)$
algorithm based on work of Cabello et al.~\cite{CC07} and Demaine et
al.~\cite{DHM07}.

\section{FPT algorithm to find minimum $k$-way cuts of
bounded size}

We want to find a minimum $k$-way cut of size at most $s$. Assuming
that a given graph is connected, the problem can only be feasible if
$k\leq s+1$. In the spirit of FPT, we are going to $O^*$ to denote $O$
assuming that the relevant parameters are constant. We will solve
the problem in $O(s^{s^O(s)}n^2)=O^*(n^2)$ time.

\subsection{The powercut problem}
To solve the problem $k$-way cut problem inductively, we are going to address a
more general problem:
\begin{definition} The {\em powercut} problem takes as input a triple
  $(G,T,s)$ where $G$ is a connected graph, $T\subseteq V(G)$
  a set of terminals, and $s$ a size bound parameter.
  For every $j\leq s+1$ and for every partition $P$ of $T$ into $j$ sets,
  some of which may be empty, we want a minimal $j$-way cut $C_{j,P}$ of $G$
  whose sides partitions $T$ according to $P$ but only if there is such a
  feasible cut of size at most $s$. The powercut is thus a cut family $\C$
  containing the cuts $C_{j,P}$, each of size at most $s$.
\end{definition}
It may be that we for different partitions $P$, $P'$ get the
same edge cut $C_{j,P}=C_{j,P'}$. Often we will identify a powercut with
its set of distinct edge cuts.

Note that if $|T|$ and $s$ are bounded, then so is the total size
of the power cut. More precisely,
\begin{observation}\label{obs:edge-bound}
The total number of edges in a power cut is bounded  by
$s\sum_{j=2}^{s+1}(j^{|T|}/j!)<(s+1)^{|T|+1}$.
\end{observation}
We will show how to solve the powercut problem in quadratic time
when $|T|,k,s=O(1)$. To solve our original problem, we solve
the powercut problem with an empty set of terminals $T=\emptyset$. In
this case, for each $j\leq s+1$, we only have a single trivial
partition $P_j$  of $T$ consisting of $j$ empty sets, and
then we return $C_{k,P_k}$.

We can, of course, also use our powercut algorithm to deal with cut
problems related to a bounded number of terminals, e.g., the $p$-pair
cut problem which for $p$ pairs $\{(s_1,t_1),...,(s_p,t_p)\}$ ask for
a minimum cut that splits every pair, that is, for each $i$, the cut
separates $s_i$ from $t_i$. If there is such a cut of size at most
$s$, we find it with a powercut setting
$T=\{s_1,...,s_p,t_1,...,t_p\}$ and $k=s+1$. In the output powercut
family $\{C_{j,P}\}$, we consider all partitions $P$ splitting every
pair, returning the minimum of the corresponding cuts. Such problems
with a bounded number of terminals and a bounded cut size, but no
restrictions on the number of components, are easier to solve
directly, as done in many cases by Marx \cite{Mar06} even for vertex
cuts. However, for vertex cuts, Marx \cite{Mar06} proved that the
$k$-way cut problem is W[1] hard. Hence the hardness is not in
splitting of a bounded set of terminals, but in getting a certain
number of components. With our powercut algorithm, we show that
getting any specified number of components is feasible with size bounded
edge cuts.

Below we will show how to solve the power cut problem recursively in
quadratic time, let $T_0$ be the initial set of terminals, e.g.,
$T_0=\emptyset$ for the $k$-way cut problem. We now fix
\begin{equation}\label{eq:t}
t=\max\{2s,|T_0|\}.
\end{equation}
In our recursive problems we will never have more than $t$ terminals.
The parameter $s$ will not change.

\paragraph{Identifying vertices and terminals}
Our basic strategy will be to look for vertices that can be identified
while preserving some powercut. To make sense of such a statement, we specify
a cut as a set of edges, and view each edge as having its own identity
which is preserved even if its end-points are identified with other
vertices. Note that when we identify vertices $u$ and
$v$, then we destroy any cut that would split $u$ and $v$. However,
the identification cannot create any new cuts. We say that $u$ and $v$
are {\em identifiable\/} if they are not separated by any cut of some
powercut $\C$. It follows that if $u$ and $v$ are identifiably, then
$\C$ is also a powercut after their identification, and then every
powercut $\C'$ after the identification is also a powercut before the
identification.
Since loops are irrelevant for minimal cuts, identifying the end-points
of an edge is the same as contracting the edge. Therefore, if the
end-points of an edge are identifiably, we say the edge is {\em contractible}.

We do allow for the case of identifiable terminals $t$ and $t'$ that
are not split by any cut in some powercut $\C$. By definition of
a powercut, this must imply that there is no feasible cut of size at most
$s$ between $t$ and $t'$.

Often we will identify many vertices. Generalizing the above notion,
we say a set of vertex pairs is {\em simultaneously identifiable\/} if there
is a powercut that does not separate any of them. This implies
that we can identify all the pairs while preserving some powercut.

Note that we can easily have cases with identifiable vertex pairs that
are not simultaneously identifiable, e.g., if the graph is a path of
two edges between two terminals, then either edge is contractible, yet
they are not simultaneously contractible.

\paragraph{Recursing on subgraphs}
Often we will find identifiable vertices recursing via a subgraph $H\subseteq G$. If $\C$ is a powercut of $G$, then $\C|H$
denotes $\C$ {\em restricted\/} to $H$ in the sense that each cut
$C\in\C$ is replaced by its edges $C\cap E(H)$ in $H$, ignoring cuts
that do not intersect $H$.

\begin{lemma}\label{lem:recurse-identify} Let $H$ be a connected subgraph of $G$. Let $S$ be
the set of vertices in $H$ with incident edges not in $H$. Define $T_H=
S\cup (T\cap V(H))$ to be the terminals of $H$. Then each powercut
$\C_H$ of $(H,T_H,s)$ is the restriction to $H$ of some powercut $\C$
of $(G,T,s)$. Hence, if pairs of vertices
are simultaneously identifiable in $(H,T_H,s)$, then they are
also simultaneously identifiable in $(G,T,s)$.
\end{lemma}
\begin{proof} Since the pairs are simultaneously identifiable in
  $(H,T_H,s)$, there is a powercut $\C_H$ of $(H,T_H,s)$ with no cut
  separating any of the pairs. Now consider a powercut $\C$
  of $(G,T,s)$, and let $C$ be any cut in $\C$. Then $H\setminus C$
  has a certain number $j\leq s+1$ of components inducing a certain
  partition $P$ of $T_H$. In $C$ we now replace $C\cap H$ with
  $C_{j,P}$ from $\C_H$, denoting the new cut $C'$. Since $C_{j,P}$ is
  a minimal, this can only decrease the size of $C$.  It is also clear
  that $G\setminus C$ and $H\setminus C$ have the same number of
  components inducing the same partition of $T$. This way we get a
  powercut $\C'$ of $(G,T,s)$ such that $\C'|H=\C_H$. In
  particular it follows that our pairs from $H$ are simultaneously
  identifiable in $(G,T,s)$.
\end{proof}

\subsection{Good separation}
For our recursion, we are going to look for good separations as
defined below.  A {\em separation\/} of the graph $G$ is defined via
an edge partition into two connected subgraphs $A$ and $B$, that is,
each edge of $G$ is in exactly one of $A$ and $B$. We refer to $A$ and
$B$ as the {\em sides\/} of the separation. Let $S$ be the set of
vertices in both $A$ and $B$. Then $S$ {\em separates\/} $V(A)$ from $V(B)$ in
the sense that any path between them will intersect $S$. Contrasting
vertex cut terminology, we include $S$ in what is separated by
$S$. In order to define a good separation, we fix
\begin{equation}\label{eq:pq}
p=(s+1)^{t+1} \textnormal{ and }q=2(p+1).
\end{equation}
Then $p$ is the upper bound from Observation
\ref{obs:edge-bound} on the total number of edges in a power cut with
at most $t$ terminals.  The separation is {\em good\/} if $|S|\leq s$
and both $A$ and $B$ have at least $q$ vertices.

Suppose we have found a good separation. Then one of $A$ and $B$
will contain at most half the terminals from $T$ because $|T| \leq 2s$. Suppose
it is $A$. Recursively we will find a powercut $\C_A$ of $A$
with terminal set $T_A=S\cup (T\cap V(A))$ as in Lemma \ref{lem:recurse-identify}. Finally in $G$ we contract all edges from $A$ that are not in the
powercut $\C_A$.

For the validity of the recursive call, we note that $|T_A|\leq
s+t/2\leq t$. The last inequality follows because $t\geq 2s$. For the
positive effect of the contraction, recall that $A$ has at least
$q=2(p+1)$ vertices where $p$ bounds the number of edges in $\C_A$. We
know that $A$ is connected, and it will remain so when we contract the
edges from $A$ that are not in $\C_A$.  In the end, $A$ has at most
$p$ non-contracted edges, and they can span at most $p+1$ distinct
vertices. The contractions thus reduce the number of vertices in $G$
by at least $|A|-p-1$.

Below, splitting into a few cases, we will look for good separations to
recurse over.

\subsection{Multiple high degree vertices}\label{sec:high-degrees}
A vertex has {\em high
  degree\/} if it has at least
\begin{equation}\label{eq:d}
d=q+s-1
\end{equation}
neighbors. Here $q$ was the lower
bound from \req{eq:pq} on the number of vertices in a side of a good
separation. Suppose that the current graph $G$ has two high degree
vertices $u$ and $v$.  In that case, check if there is a cut $D$
between $u$-$v$-cut of size at most $s$.  If not, we can trivially
identify $u$ and $v$ and recurse.

If there is a cut $D$ of size at most $s$ between $u$ and $v$,
let $A$ be the component containing $u$ in $G\setminus D$, and let
$B$ be the subgraph with all edges not in $A$.
\begin{lemma} The subgraphs $A$ and $B$ form a good separation.
\end{lemma}
\begin{proof} The set $S$ of vertices in both $A$ and $B$ are exactly
  the end-points on the $u$-side of the edges in $D$. Therefore
  $|S|\leq |D|\leq s$. We now need to show that each of $A$ and $B$
  span at least $q$ vertices. This is trivial for $B$ since $B$ contains
  $v$ plus all the $d=q+s-1$ neighbors of $v$. For the case of $A$, we
  note that $D$ can separate $u$ from at most $s$ of its neighbors.
  This means that $u$ is connected to at least $d-s=q-1$ vertices in
  $G\setminus D$, so $A$ contains at least $q$ nodes.
\end{proof}
Thus, if we have two high degree vertices, depending on the edge connectivity
between $u$ and $v$, we can either just identify $u$ and $v$, or recurse via
a good separation. Below we may therefore assume that the graph has at most
one high degree vertex.

\subsection{No high degree vertex}\label{sec:small-degrees}
Below we assume that no vertex with high degree $\geq d$---c.f.\ \req{eq:d}.
The case of one
high degree ``apex'' vertex will later be added a straightforward
extension.

\paragraph{A kernel with surrounding layers}
We start by picking a start vertex $v_0$ and grow an arbitrary
connected subgraph $H_0$, called the {\em kernel\/} from $v_0$ such that $H_0$ contains all edges
leaving $v_0$ and $H_0$ spans $h\geq d$ vertices where $h=O^*(1)$ is a
parameter to be fixed later.  Next we pick edge disjoint
minimal {\em layers\/} $H_i$, $i=1,...,p$, subject to the following constraints:
\begin{itemize}
\item[(i)] The layer $H_i$ contains no edges from
$H_{<i}=\bigcup_{j<i} H_j$, but $H_i$ contains all other edges
from $G$ incident to the vertices in $H_{<i}$.
\item[(ii)] Each component of $H_i$ is either a {\em big component} with
at least $q$ vertices, or a {\em limited component\/} with no edge
from $G\setminus H_{\leq i}$ leaving it---if a component is both, we view it as big.
\end{itemize}
Recall that a powercut $\C$ can have at most $p$ edges.  This means
that there must be at least one of the $p+1$ edge disjoint subgraphs
$H_i$ which has no edges in $\C$. Supposing we have guessed this
$H_i$, we will find a set $F_i$ of simultaneously contractible edges
from $H_0$ (note that we mean $H_0$, not $H_i$). By definition,
$F_0=E(H_0)$. If $i>0$, condition (i) implies that $H_i$ is a cut
between $H_{<i}$ and the rest of the graph, and we will use this
fact to find the set $F_i$. Since one of the guesses must be
correct, the intersection $F=\bigcap_i F_i =\bigcap_{i=1}^p F_i$
must be simultaneously contractible. An alternative outcome will be
that we find a good separation which requires $q$ vertices on either
side. This is where condition (ii) comes in, saying that we have to
grow each component of $H_i$ until either it becomes a big component
with $q$ vertices, or it cannot be grown that big because no more
edges are leaving it.

Before elaborating on the above strategy, we note that the graphs
$H_i$ are of limited size:
\begin{lemma}\label{lem:growth}
The graph $H_{\leq i}$ has at most $hd^{i}$ vertices.
\end{lemma}
\begin{proof}
  We prove the lemma by induction on $i$. By definition $|V(H_0)|=h$.
  For the inductive step with $i>0$, we prove the more precise statement
  that layer $H_i$ has at most $d$ times more vertices
  than the vertices it contains from layer $H_{i-1}$.

  Since each layer is a cut, the edges leaving $H_{<i}$
  must all be incident to $H_{i-1}$. We will argue that each component
  of $H_i$ has at most $d$ vertices. This is trivially satisfied when
  we start, since each vertex from $H_{i-1}$ comes with its at most $d-1$
  neighbors.  Now, if we grow a component along an edge, it is because
  it has less than $q<d$ vertices, including at least one from
  $H_{i-1}$. Either the edge brings us to a new vertex, increasing the
  size of the component by 1, which is fine, or the edge connects to
  some other component from $H_{i}$ which by induction had at most $d$
  vertices per vertex in $H_{i-1}$.
\end{proof}
If the $V(G)=V(H_{\leq q})$, then $G$ has only $hd^i=O^*(1)$ vertices,
and then we can solve the powercut problem exhaustively. Below
we assume this is not the case.

\paragraph{Pruning layers checking for good separations}
Consider a layer $H_i$, $i>0$, and let $H_i^-$ be the union of the big components of $H_i$. Moreover let
$H_{<i}^+$ be $H_{<i}$ combined with all the limited components from
$H_i$.  We call $H^-_i$ the {\rm pruned layer}. Since the
limited components have no incident edges from $G\setminus H_{\leq i}$,
we note that $H^-_i$ is a cut
between $H_{<i}^+$ and the rest of the graph. The lemma below
summarizes the important properties obtained:
\begin{lemma}\label{lem:pruned-layers} For $i=1,...,p$:
\begin{itemize}
\item[(i)] Pruned layer $H_i^-\subseteq H_i$ is a cut separating $H_0$
from $G\setminus V(H_{\leq q})$. In particular, we get an
articulation point if we identify all of $H_i^-$ in a single vertex.
\item[(ii)] Each component of $H_i^-$ is of size at least $q$.
\end{itemize}
\end{lemma}
Now, we take each pruned layer $H_i^-$ separately, and order the
components arbitrarily. For every pair $A$ and $B$ of consecutive
components (of order at least $q$), we check if their edge
connectivity is at least $s$ in $G$. If not, there is a cut $D$ in
$G$ with at most $s$ edges which separates $A$ and $B$. We claim
this leads to a good separation. On the one side of the separation,
we have the component $\ol{A}$ of $G\setminus D$ containing $A$, and
on the other we have the reminder $\ol{B}$ of $G$ which includes $B$
and cut edges from $D$.  Then $\ol{A}$ and $\ol{B}$ intersect in at
most $|D|\leq s$ vertices, and both $\ol{A}$ and $\ol{B}$ have at
least $q$ vertices. Thus we get a good separation.

Below we assume that for each pruned layer, the edge connectivity
between consecutive components is at least $s$.
\paragraph{Articulation points from pruned layers}
\begin{lemma}\label{lem:layer-collapse} If there is a powercut of
  $(G,T,s)$ that does not use any edge from $H_i$, then all vertices
  in the pruned layer $H_i^-$ can be identified in a single vertex $v_i$.
\end{lemma}
\begin{proof} We are claiming that no cut $D$ from $\C$ separates any
vertices from $H_i^-$. Otherwise, since $D$ does not contain any
edges from $H_i$, the cut would have to go between components from
$H_i^-$. In particular, there would be two consecutive components of
$H_i^-$ separated by $D$. However, $D$ has at most $s$ edges, and we
already checked that there was no such small cut between any
components of $H_i^-$.
\end{proof}
Below we assume we have guessed a layers $H_i$ that is not used in
some powercut of $(G,T,s)$. Let $[H_i^-\mapsto v_i]$ denote that all
vertices from $H_i^-$ are identified in a single vertex $v_i$, which
we call the \emph{articulation point}. From Lemma
\ref{lem:layer-collapse} it follows that some powercut is preserved
in $(G[H_i^-\mapsto v_i],T[H_i^-\mapsto v_i],s)$.

Next, from Lemma \ref{lem:pruned-layers} (i) we get
that $H_{\leq i}[H_i^-\mapsto v_i]$ is a block of $G[H_i^-\mapsto v_i]$
separated from the rest by the articulation point $v_i$. As in
Lemma \ref{lem:recurse-identify}, we now
find a powercut
$\C_i$ of
\[\left(H_{\leq i}[H_i^-\mapsto v_i],\
\{v_0\}\cup (T\cap V(H_{\leq i}))[H_i^-\mapsto v_i],\ s\right).\]
The powercut $\C_i$ can be found exhaustively since $H_{\leq i}$ has
at most $hd^{i}=O^*(1)$ vertices. Since this is not a recursive
call, so it is OK if $\{v_0\}\cup (T\cap V(H_{\leq i}))[H_i^-\mapsto
v_i]$ involves $t+1$ terminals.
\begin{lemma}\label{lem:excl-i} If there is a
powercut of $(G,T,s)$ that does not use any edge from $H_i$, then
there is such a powercut which agrees with $\C_i$ on $H_{\leq i}$,
and on $H_0$ in particular.
\end{lemma}
\begin{proof}
From Lemma \ref{lem:recurse-identify} we get that $\C_i$ is the
restriction to $H_{\leq i}[H_i^-\mapsto v_i]$ of some powercut
$\C_i'$ of $(G[H_i^-\mapsto v_i],T[H_i^-\mapsto v_i],s)$.
From Lemma
\ref{lem:layer-collapse} it follows that $\C_i'$ is also a powercut
of $(G,T,s)$. Since $\C'_i$ does not contain any edges contracted in
$H_i^-$, we conclude that $\C_i$ is the restriction of $\C_i'$ to
$H_{\leq i}$.
\end{proof}
With our assumption that $H_i$ is not used in some powercut,
we get that all edges in $F_i=E(H_0)\setminus E(\C_i)$ are identifiable.
Disregarding the assumption, we are now ready to prove
\begin{lemma}\label{lem:no-high} Let $F=\bigcap_{i=1}^p F_i$ be the
set of edges from $H_0$ that are not used in any $\C_i$, $i=1,...,p$.
The edges from $F$ are simultaneously contractible.
\end{lemma}
\begin{proof} Given any powercut $\C$ of $(G,T,s)$, since it
has at most $p$ edges, we know there is come $i\in\{0,...,p\}$
such that $\C$ does not use any edge
from $H_i$. If $i$ is $0$, this means all edges from $H_0$ are
contractible. For any other $i$, the claim follows from Lemma \ref{lem:excl-i}.
\end{proof}
With Lemma \ref{lem:no-high} we contract all edges from $H_0$ that are
not in some $\C_i$. From Observation \ref{obs:edge-bound}, we know
that each $\C_i$ involves at most $(s+1)^{t+2}$ edges, so combined
they involve at most $p(s+1)^{t+2}$ un-contracted edges, spanning at
most $p(s+1)^{t+2}+1$ distinct vertices.  As the initial size for $H_0$,
we start with
\begin{equation}\label{eq:h}
h=2(p(s+1)^{t+2}+1)
\end{equation}
vertices.  Therefore, when we
contracting all edges from $H_0$ that are not in some $\C_i$,
we get rid of half the vertices from $H_0$.

\subsection{A single high degree ``apex'' vertex}
All that remains is to consider the case where there is a single
high degree vertex $r$ with degree $\geq d$---c.f.\ \req{eq:d}.
We are basically going to run the reduction for no high degree from
Section~\ref{sec:small-degrees} on the graph $G\setminus\{r\}$, but
with some subtle extensions described below.

Starting from an arbitrary vertex
that is $v_0$ neighbor to $r$, we construct the layers $H_i$ in
$G\setminus\{r\}$. If this includes all vertices of $G\setminus\{r\}$, then
$G$ has $O^*(1)$ vertices, and then we find the powercut exhaustively.

Next we add the vertex $r$ to each layer $H_i$, including all edges
between $r$ and $H_i\setminus H_{< i}$. We denote this graph
$H^r_i$. Note that $H_0^r$ is connected since $r$ is a neighbor of $v_0$.
Also note that all the $H^r_i$ are edge disjoint like the $H_i$.

After the addition of $r$, for $i>0$, we turn any limited component involving
$r$ big. More precisely, in $H^r_i$ we say that a component is {\em
  big\/} if it is has $q$ or more vertices or if it contains
$r$. The remaining components are {\em limited}. Removing all other
limited components from $H^r_i$ we get the {\em pruned
  layer\/} $H^{r-}_i$. Similarly, we have the graph $H^{r+}_{<i}$
which is $H^{r}_{<i}$ expanded with the limited components from
$H^{r}_i$. Corresponding to Lemma \ref{lem:pruned-layers}, we
get
\begin{lemma}\label{lem:pruned-layers-r} For $i=1,...,p$:
\begin{itemize}
\item[(i)] The vertices from the pruned layer
$H^{r-}_{i}$ form a vertex separator in $G$ between $H_0^r$ and
$G\setminus H^{r}_{\leq q }$. In particular, we get an articulation
point if we identify $H^{r}_{i}$ in a single vertex.
\item[(ii)] Each component of $H_i^{r-}$ which does not contain $r$ has
at least $q$ vertices.
\end{itemize}
\end{lemma}
\begin{proof} Above (ii) is trivial. Concerning (i),
we already have from Lemma \ref{lem:pruned-layers} that the edges from $H^-_i$ provide a cut of
$G\setminus \{r\}$ between $H_0$ and $G\setminus \{r\} \setminus H_{\leq q }$.
The vertices in $H^-_i$ provide a corresponding vertex separation
in $G\setminus\{r\}$. When adding $r$ to the graph and to the separation,
we get a vertex separation $V(H^{r-}_i)$ between $H_0$
and $G\setminus H^r_{\leq q}$.
\end{proof}
Now, as in Section \ref{sec:small-degrees}, we order
the components of $H^{r-}_i$ arbitrarily, and check if the edge connectivity
between pairs of consecutive components is at least $s$ in $G$.
If not, we claim there is a
good separation. Let $D$ be a cut in $G$ of size at most $s$ between two
components $A$ and $B$ of $H^{r-}_i$. If $A$ and $B$ both have at
least $q$ vertices, then we have the same good separation as in
Section \ref{sec:small-degrees}.  Otherwise, one of them, say $B$
involves $r$. In this case we have an argument similar to that used for
two high degree vertices in Section \ref{sec:high-degrees}. On
one side of the good separation, we have the component $\ol{A}$ of
$G\setminus D$ including $A$. Clearly it has at least $|V(A)|\geq q$
vertices. The other side $\ol{B}$ is the rest of $G$ including
$B$ and the cut edges from $D$. Then $\ol{B}$ includes the neighborhood
of all vertices in $B$ including all neighbors of $r$, so $\ol{B}$
has at least $d+1>q$ vertices. Below we assume that we did
not find such a good separation.

We now continue
exactly as in Section \ref{sec:small-degrees}. For $i=1,...,q$ we
identify the vertices of $H^{r-}_i$ in a vertex $v_i$ which becomes
an articulation point, and then we find a powercut $\C_i$ of
\[\left(H^r_{\leq i}[H_i^{r-}\mapsto v_i],\
\{v_0\}\cup (T\cap V(H^r_{\leq i}))[H_i^{r-}\mapsto v_i]),\ s\right).\]
Corresponding to Lemma \ref{lem:no-high}, we get
\begin{lemma}\label{lem:no-high-r} The
edges from $H_0^r$ that are not used in any $\C_i$, $i=1,...,p$
are simultaneously contractible.
\end{lemma}
As in Section \ref{sec:small-degrees}, we conclude that we get
at most $p(s+1)^{t+2}$ un-contracted edges from
Lemma \ref{lem:no-high-r} and
they span at most $p(s+1)^{t+2}+1$ distinct vertices.
As the initial size for $H_0^r$, we start with $h+1=2(p(s+1)^{t+2}+1)$
vertices. Then the contractions of Lemma \ref{lem:no-high-r} allows
us to get rid of at least half the $h$ vertices in $H_0^r$.

\subsection{Analysis and implementation}
We are now going to analyze the running time including some
implementation details of the above recursive algorithm, proving a time bound of
\begin{equation}\label{eq:T}
T(n)=O\left(s^{t^{\,O(t)}}n^2\right).
\end{equation}
First, we argue that we can assume sparsity with at most $O(sn)$ edges.
More precisely, if the graph at some point has $m\geq 2sn$ edges,
as in \cite{NI92} we find $s$ edge disjoint maximal spanning forests
If an edge $(v,w)$ is not in one of these spanning forests, then
$v$ and $w$ are $s$ edge connected. We can therefore contract all
such outside edges, leaving us with at most $sn\leq m/2$ edges. This
may also reduce the number of vertices, which is only positive. The
overall cost of this process is easily bounded by  $O(sn^2)$.

In our analysis, for simplicity, we just focus on the case with no
high degree vertices from Section~\ref{sec:small-degrees}. When we
look for good separations, we check if the edge connectivity between
two vertex sets is $s$.  As we saw above, the graph can be assumed to
have at most $2ns$ edges, so this takes only $O(s^2n)$ time
\cite{FF65} including identifying a cut with $s$ edges if it exists.
The number of such good
separation checks is limited by the total number of components in all
the layers $H_i$, and for each layer, this is limited by the number of
vertices. Thus, by Lemma \ref{lem:growth}, we have at most
$\sum_{i=1}^p hd^{i}<2hd^p$ good separation checks, each of which
takes $O(s^2n)$ time. With $t\geq 2s$, $p=(s+1)^{t+1}$, $q=2(p+1)$,
$d=q+s-1$, and $h=2(p(s+1)^{t+2}+1)$---c.f. \req{eq:t}, \req{eq:pq},
\req{eq:d}, and \req{eq:h}---we get that the total time for good
separation checks is bounded by
\[O(2hd^ps^2n)=O(s^{t^{\,O(t)}}n).\]
If we do find a good separation, we recurse on
one of the sides $A$, which we know has at least
$q$ vertices. Including the separating vertices,
we know that $A$ has at most $n-q+s$ vertices. After
the recursion, we can identify all but
$q/2$ vertices in $A$. All this leads to
a the recurrence
\[T(n)\leq \max_{q\leq \ell\leq n-q+s}
O\left(s^{t^{\,O(t)}}n\right)+T(\ell)+T(n-\ell+q/2).\] Inductively this
recurrence satisfies \req{eq:T}, the worst-case being when $\ell$
attains one of its extreme values.

If we do not find a good separation, for $i=1,...,p$, we exhaustively
find a powercut of a graph with at most $hd^i$ vertices and $shd^i$ edges.  We
simply consider all the $(shd^i)^s$ potential cuts with $s$ edges, and
that is done in $O\left(s^{t^{\,O(t)}}n\right)$ total time. This reduces the
number of vertices in $H_0$ from $h$ to $h/2$, so again we get a
recurrence satisfying \req{eq:T}, completing the proof that \req{eq:T}
bounds our overall running time. In the case of the $k$-way cut
problem, we start with no terminals.  Then $t=2s$, and then our
running time is bounded by $O\left(s^{s^{O(s)}}n^2\right)$.

\section{Planar graphs and bounded genus graphs}\label{sec:planar}
We now present a simple algorithm for the planar case. We need
several known ingredients.
First, since a planar graph always has a vertex of degree at most 5, we
get a $k$-way cut of size at most $5(k-1)$ if we $k-1$ times
cut out the vertex of current smallest degree. Thus
we have
\begin{observation}\label{obs:triv-cut}
A simple planar graph has a $k$-way cut of size at most $5(k-1)$.
\end{observation}
The same observation was used in the previous slower algorithms for
$3$-way cuts \cite{He91,HS85}.

Our new $k$-way cut algorithm applies to planar graphs
with parallel edges, but like our algorithm for general
graphs, it needs a bound $s$ on the size of the cuts
considered. Such a size bound will also be used for bounded genus graphs.
We will apply the algorithm to a
simple planar graph using the bound $s=5k-5$ from Observation
\ref{obs:triv-cut}.  Parallel edges will turn up as the
algorithm contracts edges in the original graph, but the size of
the minimum $k$-way cut will not change and neither will the value of $s$.

Our algorithm uses the notion of tree decompositions and tree-width. The
formal definitions are reviewed in Appendix \ref{sec:tree-width}, which
also includes the proof of the lemma below which is kind of folklore:
\begin{lemma}\label{lem:Ken} If a planar graph $H$ has tree width at most $w$,
then we can find a minimum $k$-way cut in $O(2^{O(kw)}|V(H)|)$ time.
For a general graph $H$ of tree width at most $w$, we can find a minimum $k$-way cut in $O(w^{kw}|V(H)|)$ time.
\end{lemma}

Hereafter, $n$ always means the number
of vertices of the input graph $G$. For any set $A$ of edges, we let
$G/A$ denote $G$ with the edges $A$ contracted. If $G$ is embedded,
respecting the embedding, we contract the edges from $A$ one by one,
except that loops are deleted. We need the following theorem:

\begin{lemma}[Klein \cite{Kle08}]\label{lem:Klein}
For any parameter $q$ and a planar graph $G$ with $n$ vertices,
there is an $O(n)$  time algorithm to partition the edges of $G$ into $q$ disjoint edge sets
$S_0$,...,$S_{q-1}$ such that for each $i\in[q]$, the graph
$G/S_i$ has tree width $O(q)$.
\end{lemma}

\paragraph{A planar minimum $k$-way cut algorithm} 
If a given graph is not already embedded, we embed it in $O(n)$ time using the
algorithm from \cite{planarity}. Therefore we assume that $G$ is embedded into a plane.  
To find a minimum $k$-way cut in $G$, we set $q=s+1=5k-4$ in Lemma
\ref{lem:Klein}, and apply Lemma \ref{lem:Klein} to $G$. Next, using
Lemma \ref{lem:Ken}, we compute the minimum $k$-way cut $D_i$ of
each $G/S_i$ in $O(q 2^{O(kq)}n)=O(2^{O(kq)}n)=O(2^{O(k^2)}n)$ total
time. We return the smallest of these cuts $D_i$.
\begin{theorem}\label{thm:min-cut} We can solve the $k$-way cut problem for
a simple unweighted planar graph in $O(2^{O(k^2)}n)$ time.
\end{theorem}
\begin{proof}
  Cutting after some edges have been contracted is also a cut in the
  original graph, so the cut returned by our algorithm is indeed a
  $k$-way cut. We need to argue that one of the $D_i$ is a minimal one
  for $G$. From Observation \ref{obs:triv-cut}, we know that the
  minimum $k$-way cut $D$ has at most $s$ edges, which means that it
  must be disjoint from at least one of the $s+1$ disjoint $S_i$. Then
  $D$ is also a $k$-way cut of $G/S_i$. Hence the minimum $k$-way cut
  $D_i$ of $G/S_i$ is also a minimum $k$-way cut of $G$.
\end{proof}

\paragraph{Bounded genus}
If a given graph is not already embedded, we embed it in $O(2^gn)$ time using the
algorithm from \cite{mohar}. Therefore we assume that $G$ is embedded into a surface of genus $g$.  
We now extend our planar algorithm to the bounded genus case.  From
Euler's formula, we get
\begin{observation}\label{obs:bdg}
A simple graph embedded into a surface with genus $g$ and
$n=|V(G)| \geq 6g+k$ has a minimum $k$-way cut in $G$ of size at most $6k-6$.
\end{observation}
Next we need the following
generalization of Klein's Lemma \ref{lem:Klein}:
\begin{lemma}\label{lem:kk}
For any parameter $q$ and a graph $G$ embedded into a surface of genus $g$ with $n$ vertices,
there is an $O(2^{O(g^2q)} n)$  time algorithm to partition the edges of $G$ into $q$ disjoint edge sets
$S_0$,...,$S_{q-1}$ such that for each $i\in[q]$, the graph
$G/S_i$ has tree width $O(g^2q)$.
\end{lemma}
Lemma \ref{lem:kk} with a partition time of $O(g^3 n\log n)$ follows from
\cite{CC07,DHM07}. Our time bound is better when $g,q=O(1)$. Our
proof of Lemma \ref{lem:kk} is deferred to Appendix \ref{app:genus}.
We can now proceed as in the planar case and prove:
\begin{theorem}\label{thm:genus-cut} We can solve the $k$-way cut problem for
a simple unweighted graph with genus $g$ in $O(2^{O(k^2g^2)}n)$
time.
\end{theorem}
In fact, we can also
plug Lemma \ref{lem:kk} back into Klein's original approximate TSP algorithm,
generalizing his linear time solution from the planar to the bounded genus case.

\bibliographystyle{abbrv}
\bibliography{paper}

\newpage
\appendix
{\huge\centerline{Appendix}}

\section{Tree width Bounded Case}\label{sec:tree-width}

%
%
In this section, we shall deal with the tree width bounded case.

Recall that a \DEF{tree decomposition} of a graph $G$ is a pair $(T,R)$, where
$T$ is a tree and $R$ is a family $\{R_t \mid t \in V(T)\}$ of
vertex sets $R_t\subseteq V(G)$, such that the following two
properties hold:

\begin{enumerate}
\item[(1)] $\bigcup_{t \in V(T)} R_t = V(G)$, and every edge of $G$ has
both ends in some $R_t$.
\item[(2)] If  $t,t',t''\in V(T)$ and $t'$ lies on the path in $T$
between $t$ and $t''$, then $R_t \cap R_{t''} \subseteq R_{t'}$.
\end{enumerate}

The \DEF{width} of a tree-decomposition is $\max |R_t|$ for $t \in V(T)$.
The \DEF{tree width} of $G$ is defined as the minimum width taken
over all tree decompositions of $G$. We often refer to the sets $R_t$ as
a \DEF{bags} of the tree decomposition.

We first observe that if a given graph has tree-width at most $w$, then
we can construct a tree-decomposition of width at most $w$ in $O(w^w n)$ time
by Theorem \ref{grid} below.
\begin{theorem}[\cite{bod}]\label{grid}
For any constant $w$,
there exists an $O (w^w n)$ time algorithm that, given a graph $G$,
either finds a tree-decomposition of $G$ of width $w$ or concludes
that $G$ has tree width at least $w$. For a planar graph, the time
complexity can be improved to $O(2^w n)$.
\end{theorem}
Thus we just need to prove the following:


\begin{quote}
Given a tree-decomposition of width at most $w$, and for any fixed
$k$, there is an $O(w^{kw} n)$ time algorithm to find a
minimum $k$-way cut.
\end{quote}


We first observe that each graph with tree width $w$ has a vertex of degree at most $w$. 
Thus we get a $k$-way cut of size at most $kw$ if we $k-1$ times cut the vertex of current smallest degree. 
 
We follow the approach in \cite{AP89}. In fact, our proof is almost identical to
that in \cite{AP89}. So we only give sketch of our proof. 

The dynamic programming approach of Arnborg and Proskurowski \cite{AP89} assumes that $T$
is a rooted tree whose edges are directed away from the root.
For $t_1t'_1\in E(T)$ (where $t_1$ is closer to the root than $t'_1$), define
$S(t_1,t'_1)=R_{t_1}\cap R_{t'_1}$ and $G(t_1,t'_1)$ to be the induced subgraph of $G$
on vertices $\bigcup R_s$, where the union runs over all nodes of $T$ that are
in the component of $T-t_{1}t'_1$ that does not contain the root. The algorithm of
Arnborg and Proskurowski starts at all the leaves of $T$ and then
we have to compute the following:

\begin{quote}
For every $t_{1}t'_1\in E(T)$ (where $t_1$ is closer to the root
than $t'_1$), we compute the powercut with input
$(G(t_1,t'_1),S(t_1,t'_1),kw)$.
\end{quote}

It is clear that we can compute the powercut in each leaf in
time $O(w^k)$ by brute force. Given that we have computed the
powercut for all the children of $t'_1$ (i.e, for every children
$t''_1$ of $t'_1$, we have computed the powercut
$(G(t'_1,t''_1),S(t'_1,t''_1),kw)$), we have to compute the powercut
with input $(G(t_1,t'_1),S(t_1,t'_1),kw)$.

From each children bag $R_{t''_1}$ of $R_{t'_1}$, we have at most
$w^k$ information to be taken account when we work on the bag
$R_{t'_1}$.
However, some information can be
merged.

\smallskip

(1) If $A,B,S$ are a separation in $G$ (i.e, $A \cap B=S$), and the
powercuts are computed with inputs both $(A,S,kw)$ and
$(B,S,kw)$, then, in $O(w^k)$ time, we can compute the powercut with
input $(G,S,kw)$.

\smallskip

Since the information for $A$ and $B$ can be easily combined, thus (1) follows.

\medskip

By (1), if there are two children $t_2,t_3$ of $t'_1$ such that
$S(t'_1,t_2)=S(t'_1,t_3)$, then we can combine the powercuts (with
inputs $(G(t'_1,t_2),S(t'_1,t_2),kw)$ and
$(G(t'_1,t_3),S(t'_1,t_3),kw)$) from them.

Since there are at most $2^w$ different
subsets of $R_{t'_1}$,
thus from the children of $t'_1$, we have at most $w^k \times 2^w$
information in total to take into account.

%
Therefore 
in time $O(2^w w^{k} w^{k})=O(w^{kw})$, we can compute the powercut with input
$(G(t_1,t'_1),S(t_1,t'_1),kw)$.

We keep working from the leaves to the root. At the root, we can
pick up a minimum $k$-way cut. Since there are at most $n$ pieces in
the tree-decomposition $(T,R)$, in $O(w^{kw}n)$ time, we can compute
a minimum $k$-way cut in $G$.
\qed

\medskip

We now prove the planar case in Lemma \ref{lem:Ken}.
We follow the above proof of Lemma \ref{lem:Ken} for the general case. We shall show that
the running time in the above proof can be improved to $O(2^{kw} n)$
when an input graph $G$ is planar.

Two points in the above proof can be improved when an input graph
$G$ is planar.

First, in the above proof for the general case, for any $t$, we need
to consider $2^{|R_t|}$ partitions of $R_t$ when we take the
information of the children of $t$ into account (by (1), if there are more than 
$2^{|R_t|}$ children of $t$, we can merge some information from the children of $t$). 
But when a given graph
is planar, since $G$ does not have a Kuratowski graph, i.e, either a
$K_{3,3}$-minor or a $K_{5}$-minor, we can improve the bound
$2^{|R_t|}$ to $3|R_t|^3$ as follows:

We are interested in the number of children of $t$ that share at
least three vertices with $R_{t}$. There are at most $w^3$ choices
to choose three vertices. But on the other hand, if some three
vertices of $R_t$ are attached to three children of $t$, we can get a
$K_{3,3}$-minor. This is because, for each such a child $t'$ of $t$,
a subgraph of $G$ induced by $\bigcup R_s$, where the union runs
over all nodes of $T$ that are in the component of $T-tt'$ that
contains $t'$, is connected (otherwise, we can ``split'' the
children $t'$). Thus we can easily find a $K_{3,3}$-minor by
contracting each such a connected subgraph into a single vertex.
This implies that there are at most $2w^3$ children that share at
least three vertices with $R_{t}$.

For the children of $t$ that share at most two vertices with $R_t$,
we just need to take $w+w^2$ subsets of $R_t$ into account by (1). This implies
that, for each node $t$ of $T$, we only need to consider at most
$3w^3$ partitions of $R_t$ when we take the information of the
children of $t$ into account, as claimed.

Second, since $S(t'',t)$ is a vertex-cut in a planar $G$, where $t''$ is a parent of $t$, a minimal
vertex cut in $S(t'',t)$ consists of
closed curve(s) $C$ in a plane.
If $C$ does not contain a vertex $v$ in $S(t'',t)$, then either $v$
does not have any neighbor in $G(t'',t)-R_{t''}$, in which case, $v$
does not have to be in $S(t'',t)$, or $v$ does not have any neighbor
in $R_{t''}-S(t'',t)$, in which case, again, $v$ does not have to be
in $S(t'',t)$ either. Thus we may assume that $C$ consists of all the
vertices in  $S(t'',t)$. In addition, all the vertices in $G(t'',t)$
is contained inside some curve(s) in $C$. Given a cyclic order of
$S(t'',t)$ along each curve in $C$, by the planarity, we cannot have
a ``crossed'' partition, i.e, there are no four vertices
$s_1,s_2,t_1,t_2$ in the clockwise order along a curve in $C$ such that both
$s_i$ and $ t_i$ are contained in the same component in $R_t-D$ for
$i=1,2$, where $D$ is a cut. If $C$ consists of more than two curves, we can apply 
this argument separately. 
Thus it follows that when we consider
the powercut with input $(G(t'',t),S(t'',t),kw)$, we only need to
consider at most $2^{kw}$ ways to partition $S(t',t)$ into at most $k$
parts. So when we do exhaustive search for each bag $R_t$, we need to consider only $2^{kw}$ ways for the planar case (in contrast, 
we need to consider $w^k$ ways for the general case, as above). In summary:


\begin{quote}
Given a tree-decomposition of width at most $w$ for a planar graph,
and for any fixed $k$, there is an $O(2^{kw} n)$ time algorithm to
find a minimum $k$-way cut.
\end{quote}


This completes the proof of Lemma \ref{lem:Ken}.

\section{Bounded genus}\label{app:genus}
In this section, we prove Lemma \ref{lem:kk}.
We assume that the graph is
already embedded in a surface of minimal genus $g$. Thus
it
suffices to prove
\begin{lemma}\label{lem:kk1}
For any parameter $q$ and a graph $G$ embedded into a surface of Euler genus $g$ with $n$ vertices,
there is an $O(2^{g^2q} n)$  time algorithm to partition the edges of $G$ into $q$ disjoint edge sets
$S_0$,...,$S_{q-1}$ such that for each $i\in[q]$, the graph
$G/S_i$ has tree width $O(g^2q)$.
\end{lemma}

All the arguments in Theorem 3.3 in \cite{DHM07}, except for
finding a shortest non-contractible cycle in $G$, can be implemented in linear time.
This expensive part needs $O(g^{5/2} n^{3/2} \log n)$ time.

Thus we shall just give a linear time algorithm for this only super linear subproblem.

In order to prove Lemma \ref{lem:kk1}, we need some definitions.

Let $R$ be an embedding of $G$ in a surface $S$. Recall that a
surface minor is defined as follows. For each edge $e$ of $G$, $R$
induces an embedding of both $G\backslash e$ and $G/e$. The induced
embedding of $G/e$ is always in the same surface, but the removal of
$e$ may give rise to a face which is not homeomorphic to a disk, in
which case the induced embedding of $G\backslash e$ may be in
another surface (of smaller genus). A sequence of contractions and
deletions of edges results in a $R'$-embedded minor $G'$ of $G$, and
we say that the $R'$-embedded minor $G'$ is a \DEF{surface minor} of
$R$-embedded graph $G$.

A graph $G$ embedded in a surface $S$ has \DEF{face-width} or
\DEF{representativity} at least $l$, if every non-contractible
closed curve in the surface intersects the graph in at least $l$
points. This notion turns out to be of great importance in the graph
minor theory of Robertson and Seymour
and in topological graph theory, cf.~\cite{MT01}.
Let $C$ be a non-contractible cycle. We say that a cycle $C'$ is \DEF{homotopic} to $C$
if $C'$ can be continuously deformed into $C$ in the surface.

An embedding of a given graph is \DEF{minimal of face-width $l$}, if
it has face-width $l$, but for each edge $e$ of $G$, the face-width
of both $G\backslash e$ and $G /e$ is less than $l$. It is known that a
graph $G$ has an embedding in the surface $S$ with face-width at
least $l$ if and only if $G$ contains at least one of minimal
embeddings of face-width $l$ as a surface minor, see \cite{MT01}.

%
%
Theorems 5.6.1 and 5.4.1 in \cite{MT01} guarantee the following:

\begin{theorem}
\label{mini} \showlabel{mini} A minimal embedding of face-width $l$
in a surface of Euler genus $g$ has at most $N=N(g,l)$
vertices, where the integer $N$ depends on $g$ and $l$ only.
\end{theorem}

In \cite{KM08}, the following linear time algorithm is given.

\begin{theorem}
\label{linmin}\showlabel{linmin} Suppose $G$ has an embedding of
face-width at least $l$ in a surface $S$ of Euler genus $g$. Then
there is an $O(2^{gl} n)$ time algorithm to detect one of the minimal
embeddings of face-width $l$ in the surface $S$ as a surface minor
in $G$.
\end{theorem}

We also need the following result in \cite{KM08}.
\begin{theorem}
\label{short}\showlabel{short}
For each surface $S$ with Euler genus $g$ and a given integer $l$, there is an $O(2^{gl}n)$ time algorithm to decide,
for a graph $G$ embedded in the surface $S$, if the embedding of $G$ has
face-width at least $l$. Furthermore, if face-width is at most $l$, then
there is an $O(2^{gl}n)$ time algorithm to find a shortest non-contractible curve.
\end{theorem}

Finally, we need one more definition.
If $G$ is a plane 2-connected graph with outer cycle $C_1$ and another facial cycle $C_0$
disjoint from $C_1$, then we call $G$ a \DEF{cylinder} with \DEF{outer cycle} $C_1$
and \DEF{inner cycle} $C_0$. Disjoint cycles $C_1,\ldots,C_n$ in $G$ are \emph{concentric} if they bound discs
$D_{C_1} \supseteq \ldots \supseteq D_{C_n}$. The \DEF{cylinder-width} of $G$ is the largest integer $q$ such that $G$ has $q$
pairwise disjoint concentric cycles $C_1,\dots,C_{q}$ with the outer cycle $C_1$.

We are now ready to prove Lemma \ref{lem:kk1}.

{\it Proof of Lemma \ref{lem:kk1}.}


We just follow the proof of Theorem 3.3 in \cite{DHM07}. As mentioned above, the only super linear subproblem
is to find a shortest non-contractible cycle in \cite{DHM07}.

We now show how to avoid finding the shortest non-contractible cycles.

The proof of Theorem 3.3 in \cite{DHM07} needs to find the shortest non-contractible cycles
for these two points:

\begin{enumerate}
\item[(a)]
We need to deal with the case that the face-width is $O(gq)$. In this case,
\cite{DHM07} (c.f Lemma \ref{lem:kk1}) needs to find a shortest non-contractible curve.
\item[(b)]
On the other hand, if the face-width of $G$ is at least $64qg$,
we need to find a  cylinder $Q$ with the outer cycle $C_1$ such that $C_1$ is a non-contractible cycle, and
the cylinder-width of $Q$ is at least $8q$. Normally, this can be found by using
a shortest non-contractible cycle.
\end{enumerate}

Having (a) and (b), the proof of Theorem 3.3 in \cite{DHM07} can be implemented in $O(2^{g^2q} n)$.

We now obtain (a) and (b) in $O(2^{gq} n)$ time and in $O(2^{g^2q} n)$ time, respectively.

First, Theorem \ref{short} gives rise to (a) in time $O(2^{gq} n)$ with $l=O(gq)$.
Suppose that face-width is at least $64gq$.
We first apply Theorem \ref{linmin} to $G$ and the surface $S$ of
Euler genus $g$, with $l=64gq$, to get one of the minimal embeddings
of face-width $l$ in the surface $S$ as a surface minor, say $R$, in
$G$. We know from \cite{MT01} that this minor $R$ has a cylinder $Q$
(as a subgraph) with the outer cycle $C_1$ such that $C_1$ is a
non-contractible cycle, and the cylinder-width of $Q$ is at least
$8q$ (see \cite{MT01}). $Q$ can be easily found because $R$ has at
most $N=N(g,l)$ vertices by Theorem \ref{mini}.
%
This completes the proof of Lemma \ref{lem:kk1}.\qed

\end{document}